\documentclass{lmcs} 
\pdfoutput=1
\usepackage[utf8]{inputenc}

\usepackage{lastpage}
\lmcsdoi{18}{3}{22}
\lmcsheading{}{\pageref{LastPage}}{}{}%
{Apr.~29,~2021}{Aug.~12,~2022}{}


\usepackage{hyperref}

\usepackage{amsmath}
\usepackage{amssymb}
\usepackage{bussproofs}
\theoremstyle{plain} 

\newcommand{\Fat}{{\mathbf{F}}_{\mathbf{at}}}
\newcommand{\F}{\mathbf{F}}
\newcommand{\IPC}{\mathbf{IPC}}
\newcommand{\Fatb}{ \mathbf{F}_{\mathbf{at}}^{\circ\bullet}  }

\begin{document}

\title[Typability and Type Inference in Atomic Polymorphism]{Typability and Type Inference in Atomic Polymorphism}

\author[M.~Clarence Protin]{M. Clarence Protin\rsuper{a}}	
\address{Rua Defensores da Liberdade 10 $2^{o}$ Esq.\\
	7050-230 Montemor-o-Novo, Portugal}	
\email{cprotin@sapo.pt}  

\author[G.~Ferreira]{Gilda Ferreira\rsuper{b}}	
\address{Universidade Aberta, 1269-001 Lisboa, Portugal;
	Centro de Matemática, Aplicações Fundamentais e Investigação Operacional~--- Faculdade de Ci\^{e}ncias da Universidade de Lisboa, 1749-016 Lisboa, Portugal}	
\email{gmferreira@fc.ul.pt}  
\thanks{The second author acknowledges the support of FCT~--- Funda\c 
	c\~ao para a Ci\^encia e a Tecnologia under the projects UIDB/04561/2020, UIDB/00408/2020 and UIDP/00408/2020, and she is also grateful to CMAFcIO~--- Centro de Matem\'{a}tica, Aplica\c{c}\~{o}es Fundamentais e Investiga\c{c}\~{a}o Operacional and to LASIGE~--- Computer Science and Engineering Research Centre (Universidade de Lisboa).}	


\begin{abstract}
  \noindent It is well-known that typability, type inhabitation and type inference are undecidable in the Girard-Reynolds polymorphic system~$\F$. It has recently been proven that type inhabitation remains undecidable even in the predicative fragment of system~$\F$ in which all universal instantiations have an atomic witness (system~$\Fat$). In this paper we analyze typability and type inference in Curry style variants of system~$\Fat$ and show that typability is decidable and that there is an algorithm for type inference which is capable of dealing with non-redundancy constraints. 
\end{abstract}

\maketitle


\section*{Introduction}

Type inhabitation, typability and type inference in Girard-Reynolds polymorphic system~$\F$ (also known as the second-order polymorphically typed lambda calculus)~\cite{gir1,Rey,gir} are known to be undecidable~\cite{lob,DR,undec}. 

\emph{Type inhabitation} is the following problem: given a type $A$, is there a term having that type? I.e.\ is there a term $M$ such that $\vdash M:A$? Via the Curry-Howard isomorphism it corresponds to asking  if a formula is provable.  

\emph{Typability} goes in the other way around: \\

\noindent (Typ) Given a term $M$ is there a type $A$ and a type environment $\Gamma$ such that $\Gamma \vdash M:A$?\\

Via the Curry-Howard isomorphism typability corresponds to asking if a construction is a proof of a formula, i.e.\ is there $A$ such that $M$ is the proof of~$A$?

 What about being able to find the actual types (and type environments)  which satisfy the condition above if the answer is affirmative? In this paper by the problem of \emph{type inference} we mean the following: \\
 
 (Ti) Given a term $M$ such that (Typ) admits an affirmative answer find a procedure to generate types $A$ and type environments $\Gamma$ such that $\Gamma \vdash M:A$.\\
 
 Having an effective procedure for (Typ) the formulation above is readily seen to be equivalent to the standard formulation.  \\

It has recently been proven~\cite{CP} that type inhabitation remains undecidable even in a very weak fragment of system~$\F$ (known as system~$\Fat$ or atomic polymorphism~\cite{Ferreira06,FerreiraFerreira13}) where only atomic instantiations are allowed.

What about typability and type inference for~$\Fat$? Is typability decidable?  Can we obtain an algorithm for type inference? In this paper we show that typability is decidable and we give a procedure for type inference.

In analyzing the typing properties of~$\Fat$, the system is considered not in its original Church style formulation (in which case they are trivial) but in Curry style variants where terms do not come with full type annotation.  

The paper is structured as follows. In the next section we recall the system involved in the present study: system~$\Fat$. 

In Section~\ref{Wells} we show that typability is decidable for~$\Fat$ in \emph{Curry style} (according to  the terminology of~\cite{undec}[p.122] for~$\F$).  In Section~\ref{Curry} we show that typability for~$\Fat$ in the \emph{Polymorphic Curry style} (in which terms contain additional polymorphic typing information)  is also decidable, the proof also furnishing a procedure for type inference restricted to \emph{redundant} typings. In Section~\ref{non_redundant} we analyze non-redundant typability, show how we can check for such typings and generate them if they exist. In Section 5 it is remarked that the results of Section 4 yield as a particular case a method for type inference for Polymorphic Curry style $\Fat$. This section also contains various other remarks and considerations about future work.

\section{System $\Fat$}\label{Fat}

The atomic polymorphic system~$\Fat$~\cite{FerreiraFerreira13} is the fragment of Girard/Reynolds system~$\F$~\cite{gir1,Rey} induced by restricting to atomic instances the elimination inference rule for~$\forall$, and the corresponding proof term constructor. System $\Fat$ was originally proposed as a natural and appealing framework for full intuitionistic propositional calculus ($\IPC$)~\cite{Ferreira06,ESF20}. $\Fat$ expresses the connectives of~$\IPC$ in a uniform way avoiding \emph{bad connectives} (according to Girard~\cite{gir}, page 74) and avoiding commuting conversions. This explains the usefulness of the system in proof theoretical studies~\cite{Fer17, Rasiowa}. For related work in the topic see~\cite{Remarks}, where the authors investigate predicative translations of~$\IPC$ into system~$\Fat$ using the equational framework presented in~\cite{TPP19} and giving an elegant semantic explanation (relying on parametricity) of the syntactic results on atomic polymorphism. Originally system~$\F$ and system~$\Fat$  were presented in \emph{Church style} where types are embedded in terms. The questions of typability and type inference, mentioned in the Introduction, are meaningful only when we are in (variants of) \emph{Curry style}, i.e., terms are untyped and the type information is kept apart. Each term has a set (which may be empty) of possible types. In what follows we present $\Fat$ in \emph{Polymorphic Curry style}.

The types/formulas in~$\Fat$ are exactly the ones of system~$\F$:
\[
A,B\,::=\,X\,|\,A\to B \,|\, \forall X. A
\]

The terms in~$\Fat$ are given by
\[
M,N\,::=\,x\,|\,MN \,|\, \lambda x.M \,|\, MX \,|\, \Lambda X. M
\]

\noindent having the following typing/inference rules

{\footnotesize{\noindent\begin{prooftree}
			
			 \AxiomC{}\LeftLabel{\tiny{\textup{$(\textup{VAR})$}}}\RightLabel{$x\in dom(\Gamma)$} \UnaryInfC{$\Gamma\vdash x:\Gamma(x) $}

			 \noLine \UnaryInfC{}\AxiomC{$\Gamma\cup \{x:A\}\vdash M:B$}\LeftLabel{\tiny{\textup{(ABS)}}} \UnaryInfC{$\Gamma\vdash\lambda x.M:A\to B$}  \AxiomC{$\Gamma\vdash M:A\to B$} \AxiomC{$\Gamma\vdash N:A$} \LeftLabel{\tiny{\textup{(APP)}}}\BinaryInfC{$\Gamma\vdash MN:B$}  \noLine
			\BinaryInfC{}\noLine\BinaryInfC{}
\end{prooftree}}}

{\footnotesize{\noindent\begin{prooftree} \AxiomC{$\Gamma \vdash M:A$}\LeftLabel{\textup{(GEN$_X$)}}\RightLabel{$X\notin FTV(\Gamma)$}
			\UnaryInfC{$\Gamma \vdash \Lambda X.M : \forall X.A$}   \AxiomC{$\Gamma \vdash M: \forall X.A$}\LeftLabel{\textup{(INST)}} \UnaryInfC{$\Gamma \vdash MY:A[Y/X]$}  \noLine
			\BinaryInfC{}
\end{prooftree}}}

\noindent where a type environment $\Gamma$ is a finite set $x_1:A_1,\ldots , x_n:A_n$ where $x_i$ are assumption (term) variables and $A_i$ are types.  We denote the set $\{x_1,\ldots,x_n\}$ by~$dom(\Gamma)$, the set $\{A_1,\ldots,A_n\}$ by~$ran(\Gamma)$ and write $\Gamma(x) = A$ whenever $x:A \in \Gamma$.
We denote the free variables in the types in~$ran(\Gamma)$ by~$FTV(\Gamma)$.

Note that, as stressed in the beginning of this section, the difference between system~$\Fat$ and system~$\F$ lies in the restriction of the universal instantiation rule to type variables.

We can drop some of the information in the terms and present $\Fat$ alternatively in \emph{Curry style} (see~\cite{undec}, page 122 in the context of system~$\F$) which differs from the \emph{Polymorphic Curry style} in the terms allowed

\[
M,N\,::=\,x\,|\,MN \,|\, \lambda x.M 
\]

\noindent having the following type inference rules

{\footnotesize{\noindent\begin{prooftree} \AxiomC{}\LeftLabel{\tiny{\textup{$(\textup{VAR})$}}}\RightLabel{$x\in dom(\Gamma)$} \UnaryInfC{$\Gamma\vdash x:\Gamma(x)$}\noLine \UnaryInfC{}\AxiomC{$\Gamma\cup \{x:A\}\vdash M:B$}\LeftLabel{\tiny{\textup{(ABS)}}} \UnaryInfC{$\Gamma\vdash\lambda x.M:A\to B$}  \AxiomC{$\Gamma\vdash M:A\to B$} \AxiomC{$\Gamma\vdash N:A$} \LeftLabel{\tiny{\textup{(APP)}}}\BinaryInfC{$\Gamma\vdash MN:B$}  \noLine
			\BinaryInfC{}\noLine\BinaryInfC{}
\end{prooftree}}}

{\footnotesize{\noindent\begin{prooftree} \AxiomC{$\Gamma \vdash M:A$}\LeftLabel{\textup{(GEN$_X$)}}\RightLabel{$X\notin FTV(\Gamma)$}
			\UnaryInfC{$\Gamma \vdash M : \forall X.A$}   \AxiomC{$\Gamma \vdash M: \forall X.A$}\LeftLabel{\textup{(INST)}} \UnaryInfC{$\Gamma \vdash M:A[Y/X]$}  \noLine
			\BinaryInfC{}
\end{prooftree}}}

 In what follows, we use $A\{Y/X\}$ to denote a type that results from $A$ by replacing some (possibly all) free occurrences of~$X$ by~$Y$. Observe that if we have $\Gamma\vdash M:A$ in \emph{Curry style} derived using the (INST) rule then we have $\Gamma\vdash M:\forall X.A\{X/Y\}$ for some type variables $X$ and $Y$.  Likewise if we have $\Gamma\vdash MY:A$ in \emph{Polymorphic Curry style} then $\Gamma\vdash M:\forall X.A\{X/Y\}$ for some type variable $X$. 

\section{Typability for the Curry system}\label{Wells}

In this Section we show that the typability problem for~$\Fat$, considering the system in \emph{Curry style}, is decidable. The strategy will be to reduce the typability problem to the corresponding problem in the simply typed lambda-calculus $\lambda^\to$ which is known to be decidable. To do this,  although not strictly necessary,  we will employ a theorem which holds more generally for Polymorphic Curry style  terms.

Fix a type variable $\circ$. Given a type $A$ of~$\Fat$ we define $\langle A \rangle$ as the type of~$\lambda^{\rightarrow}$ given
by $\langle X \rangle := \circ$,  $\langle B \rightarrow C \rangle := \langle B \rangle\rightarrow \langle C \rangle$ and $\langle \forall X. A \rangle := \langle A \rangle$.
Let $M$ be a term of~$\Fat$ in \emph{Polymorphic Curry style} (or in \emph{Curry style}). We define $[M]$ as the term of~$\lambda^{\to}$ given by~$[x]:=x$, $[MN] := [M][N]$, $[\lambda x.M] := \lambda x.[M]$,
$[MX] := [M]$, $[\Lambda X.M] := [M]$.   Finally for a type environment $\Gamma$ for~$\Fat$ we define the environment $\langle \Gamma \rangle$ such that $dom (\langle \Gamma \rangle):=dom(\Gamma)$  and $\langle \Gamma \rangle(x) := \langle \Gamma(x) \rangle$.

The following lemma is easily proved by induction on the structure of the type.

\begin{lem}\label{subsX1}
	Let $A$ be an arbitrary type and $X$ and $Y$ be type variables. Then 
        \[\langle \forall X.A \rangle~=~\langle A\{Y/X\} \rangle~=~\langle A \rangle.\]
\end{lem}

\begin{thm}\label{to simply typed} Given a term $M$ in \emph{Polymorphic Curry style}, a type environment $\Gamma$ and a type $A$ in~$\Fat$, \[~ \textup{if}~ \Gamma \vdash_{\Fat} M:A ~\textup{then}~  \langle \Gamma \rangle ~\vdash_{\lambda^{\to}} [M]:~ \langle A \rangle.\]
	
\end{thm}

\begin{proof}
	
	The proof is done by induction on the complexity of~$M$. Let $M$ be a variable $x$. Then it is typed
	with (VAR) $\Gamma ~\vdash_{\Fat} x: \Gamma(x)$ and obviously $\langle \Gamma \rangle ~\vdash_{\lambda^{\to}} x:~ \langle \Gamma \rangle(x) = ~\langle \Gamma(x) \rangle$.
	Consider an application $\Gamma \vdash_{\Fat} MN :A$. Then $\Gamma \vdash_{\Fat} M :B \rightarrow A$ and $\Gamma \vdash_{\Fat} N :B$
	for some $B$. By induction  hypothesis $\langle \Gamma \rangle ~\vdash_{\lambda^{\to}} [M] :~\langle B \rangle \rightarrow \langle A \rangle$ and $\langle \Gamma \rangle~ \vdash_{\lambda^{\to}} [N]:~ \langle B \rangle$ hence
	$\langle \Gamma \rangle~ \vdash_{\lambda^{\to}} [M][N]=[MN] :~\langle A \rangle$.
	If we have an abstraction $\Gamma \vdash_{\Fat} \lambda x.M :A \rightarrow B$ then $\Gamma \cup \{x:A\} \vdash_{\Fat} M : B$.
	By induction hypothesis $\langle \Gamma \cup \{x:A\} \rangle~ \vdash_{\lambda^{\to}} [M] :~ \langle B \rangle$ hence
	
	\[ \langle \Gamma \rangle \cup ~ \{x:~\langle A \rangle\} \vdash_{\lambda^{\to}} [M] : ~\langle B \rangle  \]
	
	Using (ABS) we get
	
	\[ \langle \Gamma \rangle ~ \vdash_{\lambda^{\to}} \lambda x.[M] :~ \langle A \rangle~\rightarrow~ \langle B \rangle  \]
	
	Which by definition of~$[\cdot]$ yields 
	
	\[ \langle \Gamma \rangle~  \vdash_{\lambda^{\to}} [\lambda x.M] :~ \langle A \rightarrow B \rangle  \]
	
	In the case of universal application we have $\Gamma \vdash_{\Fat} MX: A$. But then $\Gamma \vdash_{\Fat} M: \forall Y.A\{Y /X\}$.
	By induction hypothesis $\langle \Gamma \rangle ~\vdash_{\lambda^{\to}} [M] : ~\langle \forall Y.A\{Y /X\} \rangle ~\stackrel{\textup{Lemma}~~\ref{subsX1}}{=}$ $\langle A \rangle$ and since $[MX] = [M]$ we get
	\[ \langle \Gamma \rangle ~\vdash_{\lambda^{\to}} [MX]:~ \langle A \rangle. \]
	
	Finally the case of universal abstraction,  $\Gamma \vdash_{\Fat} \Lambda X.M: \forall X. A$ follows from $[\Lambda X.M] = [M]$ and
	$ \langle \forall X.A \rangle~ =~ \langle A \rangle$ and the induction hypothesis for~$\Gamma \vdash_{\Fat} M: A$.
	\end{proof}

\begin{cor} \label{wellsT}
	If $M$ is Curry typable then $M$ is~$\lambda^{\to}$ typable.
\end{cor}

\begin{proof}
	Immediate by Theorem~\ref{to simply typed}, noticing that being $M$ a term of~$\Fat$ in \emph{Curry style}, $[M]=M$.
\end{proof}

Note that, given a~$\lambda^{\to}$-term, the inverse implication is also valid, not only for the Curry system but for  the Polymorphic Curry  system.

\begin{cor}
	The typability problem is decidable in Curry system~$\Fat$.
\end{cor}

\begin{proof}
	From the previous note and Corollary~\ref{wellsT} we have that Curry typability is equivalent to~$\lambda^{\to}$ typability and the latter is known to be decidable. 	
\end{proof}

\section{Typability for Polymorphic Curry terms}\label{Curry}

In this section we address the more complex question of typability of~$\Fat$ presented in \emph{Polymorphic Curry style}. 

Let $\circ$ and $\bullet$ be two fixed different type variables in~$\Fat$.

We define the subsystem~$\Fat^{\circ\bullet}$ of~$\Fat$ whoses types are given by

\[A, B:= \circ | A \rightarrow B | \forall \bullet. A\]

and terms by

\[M,N := x | MN | \lambda x.M | M\bullet | \Lambda \bullet. M  \]

The type inference rules for~$\Fat^{\circ\bullet}$ are the same as the ones for~$\Fat$ except for the universal inference rules that are given by

{\footnotesize{\noindent\begin{prooftree} \AxiomC{$\Gamma \vdash M:A$}\RightLabel{\textup{(GEN$_\bullet$)}}
			\UnaryInfC{$\Gamma \vdash \Lambda \bullet.M : \forall \bullet.A$}   \AxiomC{$\Gamma \vdash M: \forall \bullet.A$}\RightLabel{\textup{(INST)}} \UnaryInfC{$\Gamma \vdash M\bullet:A$}  \noLine
			\BinaryInfC{}
\end{prooftree}}}

Given a type $A$ of~$\Fat$ we define $[A]$ as the type of~$\Fat^{\circ\bullet}$ given
by $[X] = \circ$, $[A \rightarrow B] = [A]\rightarrow [B]$ and $[\forall X .A] = \forall\bullet.[A]$.
For a type environment $\Gamma$ for~$\Fat$ we define the type environment
$[\Gamma]$ for~$\Fat^{\circ\bullet}$ such that $ dom ([\Gamma])= dom(\Gamma)$ and $[\Gamma](x) = [\Gamma(x)]$.

We define a map from the terms of~$\Fat$ to the terms of~$\Fat^{\circ\bullet}$ given by~$|x|=x$, $|MN|=|M||N|$, $|\lambda x. M|=\lambda x. |M|$, $|MX| = |M|\bullet$ and $|\Lambda X. M| = \Lambda \bullet.|M|$.

Analogously to Lemma~\ref{subsX1} we have the following result:

\begin{lem}\label{subsX}
	Let $A$ be an arbitrary type in~$\Fat$ and $X$ and $Y$ be type variables. Then
	$[A\{Y/X\}]~=~[A]$.
\end{lem}

\subsection{Reduction to~$\Fatb$}

\begin{prop}\label{to double bullet} Let $M$ be a  term in~$\Fat$. If there are a type environment $\Gamma$ and a type $A$ in~$\Fat$
	such that $\Gamma \vdash_{\Fat} M:A$ then $[\Gamma] \vdash_{\Fat^{\circ\bullet}} |M|: [A]$.
\end{prop}

\begin{proof}
	The proof is done by induction on the structure of~$M$. Let $M$ be a variable $x$. Then it is typed with (VAR) $\Gamma \vdash_{\Fat} x: \Gamma(x)$ and, since $|x| = x$, obviously $[\Gamma] \vdash_{\Fat^{\circ\bullet}} |x|: [\Gamma](x) = [\Gamma(x)]$.

	Consider an application $\Gamma \vdash_{\Fat} MN :A$. Then $\Gamma \vdash_{\Fat} M :B \rightarrow A$ and $\Gamma \vdash_{\Fat} N :B$
	for some type $B$. By induction hypothesis $[\Gamma] \vdash_{\Fatb} |M| :[B\to A]=[B] \rightarrow [A]$ and $[\Gamma] \vdash_{\Fatb} |N| :[B]$ hence, since $|MN| = |M||N|$,
	we have $[\Gamma] \vdash_{\Fat^{\circ\bullet} } |MN| :[A]$.

	If we have an abstraction $\Gamma \vdash_{\Fat} \lambda x.M :A \rightarrow B$ then $\Gamma \cup \{x:A\} \vdash_{\Fat} M : B$.
	By induction hypothesis $[\Gamma \cup \{x:A\}] \vdash_{\Fatb} |M| : [B]$ hence
	
	\[ [\Gamma] \cup \{x:[A]\} \vdash_{\Fat^{\circ\bullet}} |M| : [B]  \]

	\[ [\Gamma]  \vdash_{\Fat^{\circ\bullet} } \lambda x.|M| : [A]\rightarrow [B]  \]

	\[ [\Gamma]  \vdash_{\Fat^{\circ\bullet}  } |\lambda x.M| : [A \rightarrow B]  \]

	In the case of INST we have $\Gamma \vdash_{\Fat} MX: A$. But then $\Gamma \vdash_{\Fat} M: \forall Y.A\{Y /X\}$.
	By induction hypothesis $[\Gamma] \vdash_{\Fat^{\circ\bullet} } |M| : [\forall Y.A\{Y /X\}] =\forall \bullet.[A\{Y/X\}]\stackrel{\textup{Lemma~\ref{subsX}}}{=} \forall\bullet.[A]$ and applying $|M|\bullet =|MX|$ we get
	\[ [\Gamma] \vdash_{\Fat^{\circ\bullet} } |MX|: [A]. \]

	Finally consider the case of (GEN$_X$),  $\Gamma \vdash_{\Fat} \Lambda X.M: A = \forall X. B$.
	We have $\Gamma \vdash_{\Fat} M:B$ with $X$ not belonging to the free type variables in~$\Gamma$. By induction $[\Gamma] \vdash_{\Fatb} |M|:[B]$.
	Then since $\bullet$ does not belong to the free type variables in~$[\Gamma]$, using (GEN$_{\bullet}$) we get  $[\Gamma] \vdash_{\Fatb} \Lambda \bullet. |M|:\forall \bullet.[B] = [\forall X.B]$.
	Since $\Lambda\bullet. |M| = |\Lambda X. M|$, the result follows.
	\end{proof}

Note that we work modulo $\alpha$-equivalence for type variables, in particular we assume the name of the bound variables is always appropriately chosen. 
Since in~$\Fatb$, the type variable $\bullet$ does not occur in the scope
of any quantifier it is clear that if we substitute occurrences of~$\bullet$ in a type $A$ of~$\Fatb$  by any type variables other $\circ$, then we obtain a~$\alpha$-equivalent type $A'$ in
$\Fat$.

\begin{lem} \label{bulletsub}
	Assume that $\Gamma\vdash_{\Fatb} M:A$. Let $M'$ be the term of~$\Fat$ which results from substituting some occurrences of subterms of the form 
	$\Lambda \bullet.N$ in~$M$ by~$\Lambda X. N$ with $X$ some type variable different from $\circ$ and $N\bullet$ in~$M$ by~$NY$ with $Y$ any type variable (different occurrences of the subterms may be replaced using different type variables). Then $\Gamma \vdash_{\Fat} M':A'$ where $A'$ is~$\alpha$-equivalent to~$A$.
\end{lem}

\begin{proof}
	The proof is done by induction on the structure of~$M$.
	
	If $M$ is an assumption variable $x$ and $\Gamma \vdash_{\Fatb} x:A$ 
	then $x' = x$ and the result follows immediately since $\Gamma \vdash_{\Fatb} x :A$ implies $\Gamma \vdash_{\Fat} x :A$.

	If $M$ is~$NS$ and $\Gamma \vdash_{\Fat^{\circ\bullet}}NS:A$ then it comes from (APP).
	
	Thus we have $\Gamma \vdash_{\Fat^{\circ\bullet}}N:B\to A$ and $\Gamma \vdash_{\Fat^{\circ\bullet}}S:B$ for some type $B$.
	
	Let us write $M' = (NS)' = N' S'$ where $N'$ and $S'$ are the result of applying restrictions of the substitution defining $M'$.
	Then by induction hypothesis $\Gamma \vdash_{\Fat}N':(B\to A)'$  and $\Gamma \vdash_{\Fat}S': B'$ for some type $B'$ $\alpha$-equivalent to~$B$ and some type $(B\to A)'$ $\alpha$-equivalent to~$B\to A$. Thus $(B\to A)'$ is of the form $C\to D$ where $C$ is~$\alpha$-equivalent to~$B$, and thus $\alpha$- equivalent to~$B'$, and $D$ is~$\alpha$-equivalent to~$A$. Applying (APP) we obtain $\Gamma \vdash_{\Fat} N' S'=(NS)': D$ with $D$ $\alpha$-equivalent to~$A$.

	If $M$ is~$\lambda x.N$ then $\Gamma \vdash_{\Fatb}\lambda x.N:A\to B$ which comes from (ABS) and we have $\Gamma \cup \{x:A\}\vdash_{\Fatb}N:B$.
	Let us write $M' = \lambda x. N'$. Then by induction hypothesis $\Gamma \cup \{x:A\}\vdash_{\Fat}N':B'$ with $B'$ $\alpha$-equivalent to~$B$. But using (ABS) we obtain
	$\Gamma \vdash_{\Fat} \lambda x. N' : A\to B'$. Then the result follows since $M' = \lambda x. N'$ and $A\to B'$ is~$\alpha$-equivalent to~$A\to B$.

	If $M$ is~$N\bullet$ then $\Gamma \vdash_{\Fatb}N\bullet:A$ comes from (INST) and we have
	$\Gamma \vdash_{\Fatb}N:\forall \bullet.A$.
	
	Let $M' = N' X$ (where $X$ may be~$\bullet$). By induction hypothesis 
	$\Gamma \vdash_{\Fat}N':(\forall \bullet.A)'$ where $(\forall \bullet.A)'$ is
	$\alpha$-equivalent to~$\forall \bullet. A$. If we write $(\forall \bullet.A)' = \forall Y . B$ (where $Y$ can be~$\bullet$) then it is clear
	that $B$ must be~$\alpha$-equivalent to~$A$. Applying (INST) for~$X$ yields $\Gamma\vdash_{\Fat} N'X : B$ since $Y$ is not free in~$B$. 
	The result then follows, since $M' = N'X$ and $B$ is~$\alpha$-equivalent to~$A$.

	If $M$ is~$\Lambda\bullet.N$ then $\Gamma \vdash_{\Fatb}\Lambda\bullet .N:\forall \bullet.A$ 
	comes from (GEN$_\bullet$) and we have $\Gamma \vdash_{\Fatb}N:A$.  Let us write $M' = \Lambda X.N'$ where $X$ can be~$\bullet$ but not $\circ$.
	By induction hypothesis $\Gamma \vdash_{\Fat}N':A'$ with $A'$ $\alpha$-equivalent to~$A$. Applying (GEN$_X$) (since the proviso is trivially satisfied) yields
	$\Gamma \vdash_{\Fat}\Lambda X. N': \forall X. A'$. The result then follows since  $\forall X. A'$ is~$\alpha$-equivalent to~$\forall \bullet.A$.
	\end{proof}

\begin{cor} \label{to bullet}Given a Polymorphic Curry-term $M$, if~$|M|$ is typable in~$\Fatb$ then $M$ is typable in~$\Fat$.
\end{cor}

\begin{proof}
	Let $M$ be a term in~$\Fat$ which does not contain a subterm of the form $\Lambda \circ . N$. Then $M$ results from $|M|$ by a substitution process as in Lemma~\ref{bulletsub} and the result follows.	
	Let $M$ be a term in~$\Fat$ which contains a subterm of the form $\Lambda \circ . N$. Note that the typability of a term in~$\Fat$ does not change if we rename all the type variables appearing in the term. Thus renaming $\circ$ to a fresh type variable the result follows as above.
\end{proof}

Hence combining Proposition~\ref{to double bullet} and Corollary~\ref{to bullet}  we get

\begin{thm} \label{Fattobullet}
	A Polymorphic Curry-term $M$ is typable in~$\Fat$ if and only if~$|M|$ is typable in~$\Fat^{\circ\bullet}$.
\end{thm}

\subsection{Typability in~$\Fat^{\circ\bullet}$}

In what follows we work in~$\Fatb$.

We define the set $\mathbb{S}$ of \emph{type schemes} by

\[\sigma, \tau := \alpha | \sigma \rightarrow \tau | \forall \bullet. \sigma\]

\noindent where $\alpha$ is a type scheme variable.

A \emph{generalised environment} is a finite set of elements of the form $x:\sigma$, where $x$ is an assumption variable and $\sigma$ is a type scheme. Unlike in the usual environments $x$ may occur more than once.

By a \emph{multisystem} we mean a finite set\footnote{Or equivalently a finite sequence.} $\{e_1,\ldots,e_k\}$ such that each $e_i$ is
a multiequation\footnote{We consider equality between type schemes as usual for types.} $a_1 =\cdots=a_j$ with the $a_i$ in~$\mathbb{S}$. 
Given a generalised environment $D$ we get a multisystem $\tilde{D}$ in the following way:
let $D[x]$ be the set of all elements $\sigma_i$ in~$\mathbb{S}$ such that $x:\sigma_i \in D$. Then we form the multiequation $\sigma_1=\cdots=\sigma_k$ where
the $\sigma_i$ are all the elements in~$D[x]$. The multisystem $\tilde{D}$ is defined by taking the union of such multiequations for all $x$ in the domain of~$D$.

Let $M$ be a term of~$\Fatb$.  Assume  w.l.o.g.\ that $M$ is in \emph{Barendregt form}, that is,
the names of all the bound assumption variables are chosen to be distinct between themselves as well as distinct from the  free assumption variables\footnote{For example the term $M = (\lambda x.x)(\lambda x.x)$ is to be presented in the form  $(\lambda x.x)(\lambda y.y)$ for~$y$ an assumption variable distinct from $x$. The subterms $\lambda x.x$ and $\lambda y.y$ are considered distinct.  }.
A \emph{tagging} $\mathcal{A}$ for~$M$ is a map which assigns subterms and bound assumption variables\footnote{The bound variables are needed in the definition because in terms like $M:=\lambda x. yz$, we want to have tags not just for the subterms $y$, $z$, $yz$ and $M$ but also for~$x$.} of~$M$ to type scheme variables in such a way that distinct subterms/bound variables are assigned distinct type scheme variables. We denote the image of the tagging of a subterm/bound variable $N$ by~$\mathcal{A}(N)$. 
 Here we are using strict (typographical) equality of subterms.

\begin{defi}
	Let $M$ be a term of~$\Fatb$ and $\mathcal{A}$ a tagging for~$M$. The \emph{associated multisystem} of~$M$ is the set $\{e_N\}_{N\in \mathcal{T}}$	where $\mathcal{T}$ is the set of subterms of~$M$ which are not assumption variables and $e_N$ are multiequations defined as follows:
	
	\begin{itemize}
		\item if~$N:=SP$, $\mathcal{A}(N)=\alpha_1$, $\mathcal{A}(S)=\alpha_2$ and $\mathcal{A}(P)=\alpha_3$ then $e_N$ is~$\alpha_2=\alpha_3\to \alpha_1$ 
		\item if~$N:=\lambda x.S$, $\mathcal{A}(N)=\alpha_1$, $\mathcal{A}(x)=\alpha_2$ and $\mathcal{A}(S)=\alpha_3$ then $e_N$ is~$\alpha_1=\alpha_2\to \alpha_3$ 
		\item if~$N:=S\bullet$, $\mathcal{A}(N)=\alpha_1$ and $\mathcal{A}(S)=\alpha_2$ then $e_N$ is~$\alpha_2=\forall \bullet. \alpha_1$ 
		\item if~$N:=\Lambda  \bullet . S$, $\mathcal{A}(N)=\alpha_1$ and $\mathcal{A}(S)=\alpha_2$ then $e_N$ is~$\alpha_1=\forall \bullet. \alpha_2$ 
	\end{itemize}
	
\end{defi}

\begin{lem}\label{lemma iam} Let $M$ be a term in~$\Fatb$ such that $\Gamma\vdash_{\Fatb} M:A$ and let $\mathcal{M}$ be an associated multisystem of~$M$. Then  it is possible to replace each type scheme variable in~$\mathcal{M}$ by a type in~$\Fatb$ in such a way that the multiequations in~$\mathcal{M}$ hold.
\end{lem}

\begin{proof}
	Suppose that $\Gamma\vdash_{\Fatb} M:A$. Then $\Gamma$ defines a typing for each subterm/bounded variable $N$ of~$M$, which we denote by~$\mathcal{B}(N)$. Let $\mathcal{A}$ be a tagging for~$M$ and $\mathcal{M}$ be its associated  multisystem. Let us prove that replacing the type scheme variables $\alpha=\mathcal{A}(N)$ occurring in~$\mathcal{M}$ by~$\mathcal{B}(N)$ the multiequations (on $\Fatb$ types) hold.
	
	Let $e$ be a multiequation in~$\mathcal{M}$.
	
	\begin{itemize}
		\item[] If~$e$ is of the form $\alpha_1=\alpha_2\to \alpha_3$ then two situations may occur: i) there are terms $S$ and $P$ such that $\alpha_1=\mathcal{A}(S)$, $\alpha_2=\mathcal{A}(P)$ and $\alpha_3=\mathcal{A}(SP)$ or ii) there are terms $x$ and $S$ such that $\alpha_2=\mathcal{A}(x)$, $\alpha_3=\mathcal{A}(S)$ and $\alpha_1=\mathcal{A}(\lambda x.S)$. In case i), after the replacement, we obtain the $\mathcal{B}(S)=\mathcal{B}(P) \to \mathcal{B}(SP)$ and in case ii) we obtain $\mathcal{B}(\lambda x.S)=\mathcal{B}(x) \to \mathcal{B}(S)$. Hence the equations hold.
		
		\item[] If~$e$ is of the form $\alpha_1=\forall \bullet .\alpha_2$ then two situations may occurs: i) there is a term $S$ such that $\alpha_1=\mathcal{A}(S)$ and $\alpha_2=\mathcal{A}(S\bullet)$ or ii) there is a term $S$ such that $\alpha_1=\mathcal{A}(\Lambda \bullet . S)$ and $\alpha_2=\mathcal{A}(S)$. In case i), after the replacement, we obtain the $\mathcal{B}(S)=\forall \bullet .\mathcal{B}(S\bullet)$ and in case ii) we obtain $\mathcal{B}(\Lambda \bullet.S)=\forall \bullet .\mathcal{B}(S)$. Hence the equations hold.  \qedhere
		
	\end{itemize} 	
\end{proof}

In general given a term $M$ and a tagging $\mathcal{A}$,  a \emph{solution} to the associated multisystem $E:= \{e_N\}_{N \in \mathcal{T}}$ is an assignment $f$  of types of~$\Fatb$ to the type scheme variables in~$E$ such that the resulting set of multiequations (now in~$\Fatb$) hold.\\

Given a multiequation $e$ we denote by~$e[i]$ the $i$th element counting from the left. We denote by~$e\setminus e[i]$ the multiequation
which results from eliminating the $i$th term from $e$: if~$e$ is~$\sigma_1=\cdots=\sigma_i=\cdots=\sigma_n$ then $e\setminus e[i]$ is~$
\sigma_1=\cdots=\sigma_{i-1}= \sigma_{i+1} =\cdots=\sigma_n$. If~$e$ is a single equation, $e\setminus e[i]$ consists in eliminating the equation $e$.  Given two multiequations $e_1$ and $e_2$ we denote by~$e_1.e_2$ the result of concatenating
the two multiequations into a single multiequation in the expected way.

We consider the following rules to transform (reduce) a multisystem $F$.\\

\AxiomC{$E$}
\AxiomC{$e$}
\AxiomC{$e'$}\TrinaryInfC{$E \quad e\setminus e[i].e'$}
\DisplayProof
(Join$_{i,j}$), if~$e[i] = e'[j]$ \\ \\

\AxiomC{$E$}
\AxiomC{$e$}
\BinaryInfC{$E \quad e\setminus e[i] \quad  \sigma = \tau \quad \sigma' = \tau'$}
\DisplayProof
(Arr$_{i,j}$), if~$e[i] = \sigma \rightarrow \sigma'$ and $e[j] = \tau \rightarrow \tau'$, $i\neq j$\\ \\

\AxiomC{$E$}
\AxiomC{$e$}
\BinaryInfC{$E \quad e\setminus e[i] \quad  \sigma = \tau$}
\DisplayProof
(Quant$_{i,j}$), if~$e[i] = \forall \bullet .\sigma$ and $e[j] = \forall \bullet. \tau$, $i\neq j$\\ \\

where $E$ are the other multiequations in the multisystem $F$.


\begin{rem}\label{subst}
	Let $E$ be a multisystem whose multiequations hold for a particular substitution of type scheme variables by types in~$\Fatb$. By applying the rules above, the multisystem $E'$ obtained from $E$ is such that its multiequations still hold under the same substitution.
\end{rem}

\begin{defi} Given a multisystem $E$, a \emph{reduction sequence} is a sequence of applications of the above rules.
	A multisystem is called \emph{irreducible} if it is no longer possible to apply a rule to it.
	By a \emph{clash} we mean a multiequation that contains
	type schemes both of the form $\sigma \rightarrow \sigma'$ and $\forall \bullet. \sigma$.
\end{defi}

The following observation is immediate:

\begin{lem}\label{lemma soundness} The rules (considered in both directions) are sound. More precisely,  for any substitution of the type scheme variables for types of~$\Fatb$ the rules are sound for the standard notion of type equality.
\end{lem}

\begin{defi}
	The length of a type scheme is defined inductively as follows:
	\begin{itemize}
		\item[] $l(\alpha)=1$ 
		\item[] $l(\sigma\to \tau)=l(\sigma)+l(\tau)+1$
		\item[] $l(\forall \bullet. \tau)=l(\tau)+1$ 
	\end{itemize}
	
\end{defi}

\begin{lem}\label{finite} All reduction sequences of a multisystem are finite.
\end{lem}

\begin{proof}
	Given a multisystem $E$ we consider the sequence $\sigma_1,\ldots,\sigma_k$ of all the type schemes which appear in all the multiequations of~$E$ with any order.
	Take the length of these type schemes to obtain a sequence of natural numbers $n_1,\ldots,n_k$. We denote this sequence by~$l(E)$.
	Consider the effect of the rules on the sequences $l(E)$.  (Join$_{ij}$) simply removes an element from the sequence.
	(Arr$_{i,j}$) replaces an element $n_k$ with four natural numbers $m,o,p,q< n_k$ and 
	(Quant$_{i,j}$) replaces an element $n_k$ with two natural numbers $m,p < n_k$.
	Thus reductions starting from a sequence $\sigma_1,\ldots,\sigma_k$ terminate in at most $4^{n_1} + \cdots +4^{n_k}$ steps.
	\end{proof}

\begin{rem}
	A multiequation of an irreducible multisystem contains at most one type scheme of the form $\sigma\to\tau$ and 
	at most one type scheme of the form $\forall \bullet. \sigma$.
\end{rem}

\begin{defi}
	A \emph{resolution} of a term $M$ is an irreducible multisystem without clashes obtained by reducing the associated  multisystem  of~$M$.
\end{defi}

It is immediate that a resolution of a term $M$ only contains multiequations of the following three forms (modulo permutations of the order):

\begin{itemize}
	\item[1)] $\alpha_1=\cdots=\alpha_n$
	\item[2)] $\alpha_1=\cdots =\alpha_{n}=\forall \bullet . \sigma$
	\item[3)] $\alpha_1=\cdots =\alpha_{n}=\sigma\to \tau$
\end{itemize}

For 2) and 3) above, we call the multisystems $\alpha_1=\cdots=\alpha_n$ the \emph{body} and $\forall \bullet . \sigma$ and $\sigma\to \tau$ the \emph{head}.
We also refer to multisystems of type 1) as bodies.

\begin{defi} Given a resolution $E$ of a term $M$, its \emph{minimisation} is constructed as follows:
	choose a single type scheme variable from each multiequation of type 1) and from each body of multiequations of types 2) and 3). We say that these chosen
	type scheme variables are \emph{associated} to their respective bodies.
	Then discard all type scheme variables in the bodies of the multiequations except for the previously selected type scheme variables.
	Let $\alpha$ be a scheme variable occurring in a head of a multiequation $e$. If~$\alpha$ occurs in the body of some multiequation, then we replace it in~$e$ by the type scheme variable associated to that body.

\end{defi}

\begin{defi} Let $E$ be the minimisation of a resolution of a term $M$. Its \emph{associated digraph} $G(M)$ is constructed as follows. Its vertices are
	the type scheme variables that appear in~$E$ and we have a directed edge $(\alpha_1,\alpha_2)$  whenever there is a multiequation $e$ in~$E$ such that $\alpha_2$ is in its body and $\alpha_1$ is in its head.
	
\end{defi}

We illustrate the previous concepts with an example. 

Consider the term $\Lambda \bullet.  ((\Lambda \bullet. x_1) \bullet ) (\lambda x_1. ( x_0 (\Lambda \bullet. x_2)))$.
We tag the subterms and bound assumption variables:\\

$\Lambda \bullet.  ((\Lambda \bullet. x_1) \bullet ) (\lambda x_1. ( x_0 (\Lambda \bullet. x_2))): \alpha_0$

$((\Lambda \bullet. x_1) \bullet ) (\lambda x_1 .( x_0 (\Lambda \bullet. x_2))): \alpha_1$

$((\Lambda \bullet. x_1) \bullet ): \alpha_2$

$\lambda x_1. ( x_0 (\Lambda \bullet. x_2)): \alpha_3$

$\Lambda \bullet. x_1: \alpha_4$

$x_0 (\Lambda \bullet. x_2): \alpha_5$

$\Lambda \bullet. x_2: \alpha_6$

$x_1:\alpha_7$

$x_0:\alpha_8$

$x_2:\alpha_9$\\

\noindent The associated multisystem consists of the equations:\\

$\alpha_0 = \forall \bullet. \alpha_1$

$\alpha_2 = \alpha_3 \rightarrow \alpha_1$

$\alpha_4 = \forall \bullet. \alpha_2$

$\alpha_3 = \alpha_7 \rightarrow \alpha_5$

$\alpha_8 = \alpha_6 \rightarrow \alpha_5$

$\alpha_6 = \forall \bullet. \alpha_9$

$\alpha_4 = \forall \bullet. \alpha_7$\\

\noindent Applying the rules we get the irreducible multisystem:\\

$\alpha_0 = \forall \bullet. \alpha_1$

$\alpha_4=\forall \bullet . \alpha_2$

$\alpha_7= \alpha_2 = \alpha_3 \rightarrow \alpha_1$

$\alpha_3 = \alpha_7 \rightarrow \alpha_5$

$\alpha_8 = \alpha_6 \rightarrow \alpha_5$

$\alpha_6 = \forall \bullet. \alpha_9$\\

\noindent Applying minimisation yields:\\

$\alpha_0 = \forall \bullet. \alpha_1$

$\alpha_4=\forall \bullet . \alpha_2$

$\alpha_2 = \alpha_3 \rightarrow \alpha_1$

$\alpha_3 = \alpha_2 \rightarrow \alpha_5$

$\alpha_8 = \alpha_6 \rightarrow \alpha_5$

$\alpha_6 = \forall \bullet. \alpha_9$\\

\noindent There are no clashes. The associated graph has vertices \[V = \{\alpha_0,\alpha_1,\alpha_2,\alpha_3, \alpha_4,\alpha_5,\alpha_6,\alpha_8, \alpha_9\}\] and directed edges \[E= \{ (\alpha_0,\alpha_1),(\alpha_4,\alpha_2),(\alpha_2,\alpha_1),(\alpha_2,\alpha_3),(\alpha_3,\alpha_2),
(\alpha_3,\alpha_5),(\alpha_8,\alpha_6),(\alpha_8,\alpha_5),(\alpha_6,\alpha_9) \}.\]

\begin{lem}\label{resol} A term $M$ is typable in~$\Fat^{\circ\bullet}$ if and only if it has a resolution with an associated digraph not
	containing a cycle.
\end{lem}

\begin{proof}

	Let $M$ be typable in~$\Fat^{\circ\bullet}$. By Lemma~\ref{lemma iam} and Remark~\ref{subst} we have that $M$ has a resolution.  Also, it cannot contain a cycle for this
	would imply an impossible condition on the possible types of subterms of~$M$.
	On the other hand if~$M$ has a resolution and its associated digraph $G(M)$ does not contain a cycle
	then we obtain a typing for~$M$  as follows. We tag the terminal vertices (those with no paths leaving them) of~$G(M)$ with $\circ$.
	Since there are no cycles we can tag the rest of the vertices inductively as follows. Assume that the descendants of a vertex $\alpha$ have already been tagged.
	If it has two descendants $\alpha_1$ and $\alpha_2$ tagged with $A$ and $B$ then we tag $\alpha$ with $A\rightarrow B$. If it has a single
	descendant $\beta$ tagged with $A$ then we tag it with $\forall \bullet. A$. It is easy to see by the construction of the
	associated graph that these are the only two possible situations.
	Each vertex of the associated digraph is also associated to a body in the resolution of~$M$. Substitute all type schemes in the body by the new type associated
	to the vertex. With this substitution the multiequations of the resolution continue to hold. Substituting these new values in a~$(x:\sigma)$ of the inferred typing of~$M$ 
	for each $x$ yields an environment rendering $M$ typable in~$\Fatb$.
	\end{proof}

Since it is decidable if a term $M$ has a resolution and whether its associated digraph has a cycle we get that:

\begin{lem} Typability in~$\Fatb$ is decidable.
\end{lem}

Hence, combining the previous lemma with Theorem~\ref{Fattobullet} yields:

\begin{thm}
	Typability for~$\Fat$ in \emph{Polymorphic Curry style} is decidable.
\end{thm}

The proof of Lemma~\ref{resol} yields a method of type inference for~$\Fatb$ similar to the case of simply typed calculus~\cite{Remy!appsem}[1.4.3]. The proof of Corollary~\ref{to bullet} allow us to extend such type inference strategy to a method of type inference for~$\Fat$ restricted to redundant typings.

\section{Non-redundant Typability}\label{non_redundant}

A natural question arises: when is a term $M$ of~$\Fat$ typable with non-redundant quantifiers? For example, if a term $M$ contains subterms $N(xX)$ and $N(xY)$, where $x$ is the same free or bound assumption variable (after imposing the Barendregt condition), 
then the type of~$x$ must be of the form $\forall Z.A$ and the types of~$xX$ and $xY$ must be equal. The type of~$xX$ is~$A[X/Z]$ and the type of~$xY$ is~$A[Y/Z]$. Thus $Z$ cannot occur in~$A$ and we see that $\forall Z$ is redundant. 

We have seen by Corollary 2.3 that typable Curry terms coincide with the typable terms of~$\lambda^{\rightarrow}$. 
If a term $M$ has a derivation yielding a typing $\Gamma\vdash M:A$ in which all the quantifiers in~$\Gamma$ and $A$ are redundant then by  lemma 3.11 of~\cite{undec} (which carries over from $\F$ to~$\Fat$)  $M$ has a derivation yielding a typing $\Gamma'\vdash M: A'$ in which all the redundant quantifiers in~$\Gamma$ and $A$ have been eliminated, that is, we obtain canonically a derivation in~$\lambda^{\rightarrow}$.
  Clearly  non-redundancy is very relevant for a meaningful notion of type inference in~$\Fat$. A major advantage of the Polymorphic Curry style presentation of~$\Fat$ is that the second-order information present in the terms allows us to impose fine-grained  non-redundancy conditions in the procedure for type inference.
  
  In this Section we consider $\Fat$ in Polymorphic Curry style.

\begin{defi} Given a  term $M$ in~$\Fat$  typable with a type environment $\Gamma$, we say that it is a \emph{non-redundant typing} if for any subterm of~$M$ 
of the form $NX$ or~$\Lambda X.N$ if we consider the induced typings\footnote{Given a typing $\Gamma \vdash M:A$ and a subterm $N$ of~$M$ there is a uniquely determined type environment $\Gamma'$ and type $B$ such that $\Gamma' \vdash N:B$. We call this the \emph{induced} typing on~$N$ by~$\Gamma \vdash M:A$.} $\Gamma' \vdash N: \forall Y. B$ in the first case or~$\Gamma' \vdash N: B$ in the second, then $Y$ occurs free in~$B$ in the first case or~$X$ occurs free in~$B
$ in the second. \end{defi}

Non-redundancy involves imposing a set of \emph{positive constraints} relative to variable occurrences for the types of certain subterms (which we call \emph{polymorphic subterms}). We can discard these  constraints partially (i.e.\ only impose them on a subset of such subterms) or totally without losing typability.
Typability in general requires only a set of \emph{negative constraints}  imposed by
the proviso of rule (GEN$_{X}$).   In this section we will address the problem of typability in the setting of  positive constraints on a choice of polymorphic subterms which includes  both the (empty) case corresponding to general typability  and the case of  non-redundancy above.  Given a selection $c$ of occurrences of polymorphic subterms of a term $M$ we call a typing of~$M$ which satisfies the positive constraints on these subterms \emph{$c$-non-redundant}.\\

We now define the set $\mathbb{S}'$ of \emph{type schemes} by

\[\sigma, \tau := \alpha  |    \sigma[X /Y]    |\sigma \rightarrow \tau | \forall X. \sigma\]

\noindent where $\alpha$ is a type scheme variable and $X$ and $Y$ are type variables of~$\Fat$. 
Type schemes of the form $\alpha s$ where $s$ is a (possibly empty) string of expressions of the form $[X /Y]$ are called \emph{variants} of~$\alpha$.  If a type scheme is not a variant of an~$\alpha$ then it is called
\emph{complex}.
We also write $\alpha[X_1/Y_1]...[X_n/Y_n]$ as~$\alpha[X_1...X_n/Y_1...Y_n]$ or~$\alpha[\bar{X}/\bar{Y}]$.

We consider that $(\sigma \rightarrow \tau)[X/Y] =  \sigma [X/Y]\rightarrow \tau [X/Y]$ and $(\forall Z. \sigma)[X/Y] = \forall Z. \sigma[X/Y]$ (for~$X$ and $Y$ distinct from $Z$)  and
$(\forall Z. \sigma)[X/Z] = \forall Z. \sigma$.

If  $X,Y,Z,W$ are distinct type variables we consider that $\sigma[X/Y][Z/W] = \sigma [Z/W][X/Y]$. Also we consider that $\sigma[Y/X][W/X] = \sigma[Y/X]$ and $\sigma [X/X] = \sigma$.

We define  a \emph{generalised environment} as previously. A \emph{multisystem} is also defined as previously. 

Given a term  $M$ in~$\Fat$, we define a \emph{tagging}  $\mathcal{A}$ exactly as before. We also have a \emph{variable tagging} $\mathcal{V}$ which assigns to each \emph{occurrence} of a subterm of the form $NX$ or~$\Lambda X.N$ of~$M$ a new type variable $Y$. We denote the value of~$\mathcal{V}$ sometimes by~$X_{NX}$ or~$X_{\Lambda X.N}$ when there is only one occurrence. We also use the notation $NX: \alpha$  for~$\alpha = \mathcal{A}(NX)$. 
The \emph{associated multisystem} is defined as~$F \cup \{e_N\}_{N\in \mathcal{T}}$  where $\mathcal{T}$ is
the set of occurrences of subterms  of~$M$ which are not assumption variables and $F$ is as described below.
When $N$ is an occurrence of an arrow application or an arrow abstraction $e_N$  is defined as in the previous section (note that occurrences of the same subterm of this form give rise to the same multiequations).  In the case of an occurrence of a subterm of the form $NX$ we have that $e_{NX}$ is

\[\alpha_1 =  \alpha_2^W [X/W]\]

\noindent where $\alpha_1 = \mathcal{A}(NX)$, $\alpha_2 = \mathcal{A}(N)$, $W=\mathcal{V}(NX)$ 
and $\alpha_2^W$ is a fresh type scheme variable identified uniquely by~$\alpha_2$ and $W$.
Also  $e_{\Lambda X. N}$ is

\[\alpha_1 = \forall W. \alpha_2[W/X]\]
where $\alpha_1 = \mathcal{A}(\Lambda X. N)$, $\alpha_2 = \mathcal{A}(N)$ and
$W = \mathcal{V}(\Lambda X. N)$.

Finally $F$ consists of equations $\alpha = \forall W. \alpha^W$ for all
occurences of subterms of~$M$ of the form $NX$, where $\alpha = \mathcal{A}(N)$ and $W$ is the value of~$\mathcal{V}$ for each occurrence of the subterm. 

Given a term $M$ and taggings $\mathcal{A}$ and $\mathcal{V}$ we define
a \emph{solution} to the associated multisystem $F \cup \{e_N\}_{N\in \mathcal{T}}$ 
exactly as before.

We consider a binary predicate $Occ(X,\sigma)$ over the sorts of type variables and type scheme variables which will be interpreted as meaning that $X$ must occur free in  the assignment of~$\sigma$ for any solution of the associated multisystem.\\

We associate to a term $M$ a set of formulas $\mathcal{B}(M)$ (called \emph{basic constraints}) in the following way: for all subterms of~$M$ of  the form $\Lambda X. S$ we add

 \[\bigwedge_{\alpha_i}\neg Occ(X,\alpha_i)\]

where $\alpha_i = \mathcal{A}(x)$ for the free assumption variables $x$ in~$S$.  Note that these are the constraints required for typability in general.

Given a selection $c$ for a term $M$ we define a set $\mathcal{K}_c(M) $ which includes $ \mathcal{B}(M) $ and a possibly empty set of positive constraints constructed as follows: for subterms of~$M$ chosen by~$c$ of  the form $\Lambda X. S$ we add

\[ Occ(X, \mathcal{A}(S))\]

\noindent and for subterms of~$M$ chosen by~$c$ of the form  $NX$   we add

\[ Occ(W, \alpha^W) \]

\noindent with $\alpha = \mathcal{A}(N)$ and $W$ is the  value of~$\mathcal{V}$  for that occurrence of the subterm.\\

Note that $Occ(X, \mathcal{A}(S))$  implies that $Occ(W, \mathcal{A}(S)[W/X] )$ where
$W = \mathcal{V}(\Lambda X. S)$.\\

 For other properties of substitution and the occurrence predicate (which will be used in the proofs ahead) see the \emph{occurrence axioms} listed at the beginning of Section 4.1.

Note that if we consider the term $M = \lambda x y. y(I (x X)) (I (x Y) )$ where $I = \lambda z.z$ then by
the definition of tagging we must first ensure that the bound variables are all distinct. Thus the second occurrence of the subterm $I$ would become $I' = \lambda w.w$ where $w$ is an assumption variable distinct from $z$.  We also take occasion to note that it can happen that for a term $N$ with a \emph{redundant} typing $\Gamma\vdash N:A$ we can still have that $A$ does not contain any redundant quantifiers.\\

When we write $\sigma[X/Y]$ we are assuming that $\sigma$ represents a possible type $A$ in which $X$ is free for~$Y$ in~$A$.  We call such values of~$\sigma$ \emph{adequate} for~$[X/Y]$. For a sequence of substitutions $\sigma s$ we define $A$ being adequate to~$s$ inductively.
If $s$ is a single substitution then the definition is as previously. Otherwise let $s = s' [Y/X]$.
Then a type $A$ is adequate for~$s$ if~$As'$ is adequate for~$[Y/X]$.

A partial substitution $\{X/Y\}$ is formally given by  two type variables $X$ and $Y$ and a strictly increasing (possibly empty) sequence of natural numbers $n_1,\ldots,n_k$. The interpretation is as follows. We consider the sequence of all the free occurrences of~$Y$ as they occur from the left. Then the free occurrences of~$Y$ in positions $n_1,\ldots,n_k$ in that sequence are substituted by~$X$. Obviously such a partial substitution can only be applied to types that have at least $n_k$ free occurrences of~$Y$.

We thus can extend the notion of adequacy to partial substitutions $\{Y/X\}$ and to partial substitution sequences $s$.

Note that if~$B = A[X/Y]$
with $A$ adequate for~$[X/Y]$ then $A = B\{Y/X\}$ for some partial substitution $\{Y/X\}$ with $Y$ free for~$X$ in~$B$ for the occurrences of~$X$ which are substituted. If~$A = B s$ then $B = A \mathfrak{s}$ where $\mathfrak{s}$ is some
sequence of partial substitutions involving the variables in~$s$. We can always assume that $s$ represents a \emph{simultaneous} substitution and likewise for~$\mathfrak{s}$.

Since a term has a~$c$-non-redundant typing if and only if it has a~$c$-non-redundant typing in which all the bound variables are distinct and different from the type variables that appear in the term (and if we need to rename a bound variable we can do this in such a way that this situation is preserved) we may assume that the solutions (type values) we are considering are adequate for the substitutions that appear in the associated multisystem.

\begin{lem}\label{lemma ima2} Let $M$ be a term in~$\Fat$ such that $\Gamma\vdash_{\Fat} M:A$ is a~$c$-non-redundant typing and let $\mathcal{M}$ be its 
	associated multisystem for a given $\mathcal{A}$ and $\mathcal{V}$. Then we can replace each type scheme variable in~$\mathcal{M}$ by a type in such a way that the multiequations in~$\mathcal{M}$ hold as well as the corresponding constraints $\mathcal{K}_c(M)$.
\end{lem}
\begin{proof}
	This is shown in the same way as in Lemma~\ref{lemma iam}. We need only consider two further cases. Consider a~$c$-non-redundant typing $\Gamma$ for~$M$ and let $\Gamma\vdash NX:A$ and $\Gamma\vdash N:B$  for a subterm $NX$ of~$M$. Consider a given occurrence of~$NX$ with $W$ the value of~$\mathcal{V}$ for that occurrence.
	Let $\alpha = \mathcal{A}(NX)$ and $\beta = \mathcal{A}(N)$.  Now $B$ must be of the form $\forall Y. C$ and
	we have by (INST) that $A = C[X/Y]$.  Renaming the bound variable $Y$ to~$W$ (we can assume that there are no variable clashes) we get that $B = \forall W. C[ W/Y]$. Hence $A = C[X/Y] = C[W/Y][X/W]$. Hence we put $\beta^W = C[W/Y]$ and the equation $\alpha = \beta^W[X/W]$ holds. And also $\beta = \forall W. \beta^W$.
	The other case is similar. That the constraints hold follows from the definition of~$c$ --- non-redundancy.
\end{proof}

Given a multiequation $e$ we denote by~$e[\bar{X}/\bar{Y}]$ the multiequation for which
\[e[\bar{X}/\bar{Y}][i] = e[i] [\bar{X}/\bar{Y}]\]
Given a term $M$ with a variable tagging $\mathcal{V}$ we denote by~$\circ$ a type variable
which does not occur among the variables in the substitutions of the associated multisystem of~$M$.
Let $\bar{W}$ be the sequence of all such variables and consider the substitution sequence
$[\circ /\bar{W}]$. Then for any substitution $s$ involving variables of~$\bar{W}$ we have that
$ A s [\circ /\bar{W}] = A  [\circ /\bar{W}] $. Here, when considering $A$ as part of the solution to the associated multisystem, without loss of generality we can assume that the types $A$ never contain $\circ$ either as a free or bound variable.
Given $M$ and $\mathcal{V}$ we use the notation
$\bigcirc = [\circ /\bar{W}]$.

We consider the following rules to transform a multisystem $F$. We have (Arr$_{i,j}$) as in the previous section. In addition, we have the following two rules: \\

\AxiomC{$E$}
\AxiomC{$e$}
\BinaryInfC{$E \enspace e\setminus e[i] \enspace \sigma = \tau[X/Y] \enspace  \tau = \sigma[Y/X] $}
\DisplayProof
(Quant$_{i,j}$), $e[i] = \forall X .\sigma$, $e[j] = \forall Y. \tau$, $i\neq j$\\ \\

\AxiomC{$E$}
\AxiomC{$e$}
\AxiomC{$e'$}\TrinaryInfC{$E \quad (e\setminus e[i])\bigcirc.e'\bigcirc$}
\DisplayProof
(Join$_{i,j}$),  $e[i] = \alpha s$, $e'[j] = \alpha s'$ \\ \\

In the above  $E$ are the other multiequations in the multisystem $F$.
Note that (Arr) and (Quant) are bi-directional.

\begin{defi} Given a multisystem $E$, a \emph{reduction sequence} is a series of applications of the above rules.
	A multisystem is called \emph{irreducible} if it is no longer possible to apply a rule to it.
	
\end{defi}

\begin{rem}\label{reductionFb} We can transform a reduction sequence in the above sense into a reduction sequence in~$\Fatb$ by applying the translation $[\cdot]$ of the previous section where 
	$\alpha s$ in~$\mathbb{S}'$ is translated via $[\cdot]$ to~$\alpha$  in~$\mathbb{S}$.
	
	This guarantees the existence of an irreducible multisystem for typable terms (see Lemma~\ref{finite}).
	Alternatively, we can characterise  irreducible multisystems as those which translate to irreducible multisystems in~$\Fatb$. 
		
\end{rem}

\begin{defi}
	A type scheme variable $\alpha$ in a multisystem is called \emph{terminal} if no variant of it can be expressed in terms of complex type schemes with variants of the other type scheme variables in the multisystem. Terminals are organised into \emph{groups}. Two terminals $\alpha$, $\alpha'$ belong to the same group iff $\alpha s = \alpha' s'$ for some $s,s'$.
\end{defi}

Note that the reduction rules are sound for the standard notion of equality and the terminals are preserved by reduction.

\begin{rem} To obtain a solution to a multisystem it is not enough to give values to the terminal scheme variables.  The type structure of the non-terminals is, however, determined once we have the value of the terminals, as is easily seen by Remark~\ref{reductionFb}.

	What happens is that the free variables of the non-terminals are not uniquely determined by the terminals. For example, consider the multisystem associated to the term $(xX)y$.
	Let $(xX)y:\alpha_1$, $y:\alpha_2$, $xX:\alpha_3$ and $x:\alpha_4$. The associated multisystem contains $ \alpha_3 = \alpha_2 \rightarrow \alpha_1$, $\alpha_3 = \alpha_4^W[ X/ W]$ and $
	\alpha_4 = \forall W. \alpha_4^W$ where $W = \mathcal{V}(xX)$.  Here $\alpha_1$ and $\alpha_2$ are the terminals. Consider the values $ \alpha_1 = (X \rightarrow X) $  and $\alpha_2 = \circ$. Then we have the following distinct solutions:
	\[\alpha_3 = \circ \rightarrow (X \rightarrow X )  \text{ and } \alpha_4 = \forall W. \circ \rightarrow (X \rightarrow W) \]
	\[\alpha_3 = \circ \rightarrow  (X \rightarrow X)  \text{ and } \alpha_4 = \forall W. \circ \rightarrow  (W \rightarrow W) \]
	A  \emph{structure} is an expression built up from a  ``hole'' symbol  and $\forall X.$ and $\rightarrow$.   We say that a type $A$ is \emph{built up} from types $A_1,\ldots, A_n$ if
	$A$ can be obtained by considering a certain structure in which the ``holes'' are filled with types among $A_1,\ldots,A_n$.
\end{rem}

\begin{lem} \label{parsub} Consider the associated multisystem of  a term $M$ with a~$c$-non-redundant typing.
	Given a solution consider the values of the terminals $A_1,\ldots, A_n$.	Then for each non-terminal $\alpha s$ (in particular $\alpha$) there is a partial substitution sequence $\mathfrak{s}_1,\ldots, \mathfrak{s}_m$ such that the value of~$\alpha s$ is built up from $A_1 \mathfrak{s}_1 ,\ldots, A_n \mathfrak{s}_m $, where $m \geq n$ is bounded by the length of~$M$ and the $A_i$ may be repeated in this list.
\end{lem}

\begin{rem}
Since the associated multisystem determines the structure for each non-terminal,  to specify  a solution it is enough to give values for the terminals and a sequence
$\mathfrak{s}_1,\ldots, \mathfrak{s}_n$ for each non-terminal applying to the terminals which occur in its decomposition defined by the associated multisystem.	
\end{rem}

\begin{proof}[Proof of Lemma~\ref{parsub}]
	Consider the irreducible multisystem of~$M$. Each non-terminal $\alpha$ occurs in a unique multiequation containing a complex type scheme, an arrow or quantifier element.  We call the type scheme variables that occur in this complex type scheme  the \emph{descendants} of~$\alpha$.  In this way we can associate each non-terminal with a uniquely determined non-cycling digraph $G(\alpha)$ having terminals as leaves.
	To prove the lemma proceed by induction on the structure of~$G(\alpha)$.

	If $\alpha$ is non-terminal  then its multiequation contains either $\alpha s = \forall X. \sigma$ or~$\alpha s = \sigma \rightarrow \tau$.  Note that if  a type $A$ is built up from types $A_i$ then $A\mathfrak{s}$ is built up from types $A \mathfrak{s}_i$ where $\mathfrak{s}_i$ are restrictions of~$\mathfrak{s}$.
	If $\alpha s = \sigma \rightarrow \tau$ then by hypothesis $\sigma$ and $\tau$ are built from partial substitutions on terminals. But $\alpha = (\sigma \rightarrow \tau) \mathfrak{s}$. Decomposing $\mathfrak{s}$  into partial substitution sequences $\mathfrak{s'}$ and $\mathfrak{s''}$ on 
	$\sigma$ and $\tau$ respectively we get that $\alpha = \sigma \mathfrak{s'} \rightarrow \tau \mathfrak{s''}$. 
	The case of~$\alpha s = \forall X. \sigma$ is analogous.
	\end{proof}

\begin{defi}
	An \emph{occurrence matrix} of a term $M$ is a set of formulas which consists of~$Occ(X,\alpha)$ or~$\neg Occ(X,\alpha)$ for each type variable $X$ and type scheme variable $\alpha$ occurring in the associated multisystem of~$M$.
\end{defi}	

\begin{rem}  
	If a term $M$ has a~$c$-non-redundant typing then it generates a solution to its associated multisystem which in turn generates an occurrence matrix. Notice that given this matrix we can determine all
	the occurrences or non-occurrences of the type variables of the multisystem in substitutions $\alpha s$ of type scheme variables. It is clear that $\mathcal{K}_c(M)$ has to be satisfied, since the typing is~$c$-non-redundant.
	
\end{rem}

\begin{lem} Let $M$ be a term with a~$c$-non-redundant type environment $\Gamma$ and let $\star$ be a type variable not occurring in~$\Gamma$ nor in  the associated multisystem of~$M$.  Then $M$ has a~$c$-non-redundant type environment $\Gamma'$ in which the free type variables either occur in the associated multisystem of~$M$ or are $\star$.
\end{lem}

\begin{proof}
	
	Notice that substitution sequences preserve equality. We use the following fact: let $\bar{R}$ be all the type variables in~$\Gamma$ which do not occur in the associated multisystem of~$M$. In particular the variables in~$\bar{R}$ do not occur among the variables in~$s$ in a reduction sequence. Consider any equation $ A s = B t$.  Then we have that $A s [ \star /\bar{R}]  = B t [\star / \bar{R}]$ and so  $A [ \star /\bar{R}] s = B  [\star / \bar{R}] t$. Thus if we apply $ [ \star /\bar{R}]$ to the types in~$\Gamma$ it is clear that we still obtain a typing and thus a solution to the associated multisystem. The typing is still $c$-non-redundant because we are not altering the variables belonging to the occurrence matrix of~$M$.
	\end{proof}

We call such a typing \emph{adequate}. Hence to look for~$c$-non-redundant typings for a typable term $M$ we need only look for adequate typings.  Hence to look for solutions to a multisystem we need only look for types with free variables occurring in the multisystem or equal to a distinguished type variable $\star$. We call such solutions \emph{adequate} solutions and the values attributed to the type scheme variables adequate values. The following obvious result is important:

\begin{lem}
	Given a term $M$ and adequate values for the terminals of its associated multisystem there are only a finite number of possible adequate values on the non-terminals.  
\end{lem}	

This follows from Lemma~\ref{parsub} noticing that the partial substitution sequences can involve only the  finite number of variables which occur in the associated multisystem.

Given a type $A$, its \emph{variable sequence} is the sequence, as they occur, of all free type variables in~$A$\footnote{For example for~$ A:= X \rightarrow (Y \rightarrow X)$ the variable sequence is~$X,Y,X$.  }. Note that given a solution to the associated multisystem of a term $M$ the values of terminals in the same group have the same structure.

\begin{lem}
	Consider an adequate  solution to the associated multisystem of a term $M$ which yields a~$c$-non-redundant typing. Then if we substitute the values of the terminals in a given group by types having a different structure with the same variable sequences we can obtain another solution which also yields a~$c$-non-redundant typing.
\end{lem}
\begin{proof}
	By Lemma~\ref{parsub} the non-terminals are given by partial substitution sequences on terminals and since the variable sequences of the terminals are the same after altering the structure we can use the same partial substitution sequences to obtain the values of the non-terminals.  Consider $\alpha s$ and $\alpha' s'$ belonging to the same multiequation in the reduced multisystem and let $\alpha$ and $\alpha'$ have values $A$ and $A'$. Then $A s = A' s'$. The values $A$ and $A'$ are built up from partial substitutions on the values of the terminals in~$G(\alpha) = G(\alpha')$. Let these values be~$A_1,\ldots,A_n$ and $A'_1,\ldots, A'_n$.  So to check if~$A s= A' s'$ (after changing the structure on terminals keeping the variable sequence) we only need to check if~$A_i s = A'_i s'$ for~$i=1,\ldots,n$ with the changes above. We do not have to worry about renamed bound variables in~$A$ or~$A'$ because we are working within a fixed irreducible multisystem and thus with a fixed structure for~$A$ and $A'$.
	But observe that if~$B s = B' s'$ then $C s = C' s'$ where $C$ has the variable sequence as~$B$ and $C'$ the same variable sequence as~$B'$. It follows that when we alter the structure of the values of terminals in the above way that $\alpha s = \alpha' s'$ still holds for the altered values. 
	$c$-non-redundancy follows by observing that the alteration of the structure does not affect the occurrence matrix. 
\end{proof}

Hence we can now assume that in our adequate solutions the terminals are given by types of the form $X_1\rightarrow X_2\rightarrow\ldots\rightarrow X_n$. We call this the \emph{minimal form}. We refer to the position of~$X_i$ as the $i$th place. We call $n$ the \emph{size} of the type. By the result of removing the $i$th place from a type in minimal form (of a variable sequence of size greater than 1) we mean the type $X_1\rightarrow\ldots\rightarrow X_{i-1}\rightarrow X_{i+1}\rightarrow\ldots\rightarrow X_n$.

Given an adequate solution, consider the values of terminals in a group  in minimal form. They will have the same length $n$.  A place $k$ is called \emph{essential} if after we alter the solution in such a way that this place is removed from the values of all terminals in a group the occurrence matrix changes. The values of the other type scheme variables are given by the restrictions of their respective partial substitutions sequences on terminals. 
These values are a solution because if for two terminals  $\alpha$ and $\alpha'$ their values in minimal form satisfy  $\alpha s = \alpha ' s'$ then they continue to satisfy this equation if we remove a place from both values.

Notice that removing a place cannot alter a non-occurrence constraint.

\begin{lem} If a term $M$ has a~$c$-non-redundant typing then it has a~$c$-non-redundant typing given by an adequate solution in minimal form such that the maximum length of its values on the terminals is bounded by the order  $\mathcal{O}(mn)$, where  $n$ is the number of type variables and $m$ is the number of type scheme variables in the associated multisystem of~$M$.  
\end{lem}

\begin{proof}
	Consider an adequate minimal solution for the associated multisystem of~$M$.
	Given a group of terminals assume that a place  is not essential. Then if we remove it we obtain a solution which gives a~$c$-non-redundant typing. Hence we only need to study the number of essential places on values of terminals. Observe that \emph{removing two places in a given group cannot affect the occurrence condition  for the same type variable and the same type scheme}.
	For each place $n$ consider the set $S(n)$ of affected pairs $(X,\alpha)$. Then the sets $S(n)$ are
	disjoint for different $n$.  It follows that the number of essential places in a group of non-terminals cannot exceed the product of the number of type variables and the number of type scheme variables in the associated multisystem.
	\end{proof}

The bound in the lemma can in fact be sharpened to the order of the cardinality of~$\mathcal{K}_c(M)$,  for in the proof above we are in fact only interested in not affecting certain pairs $(X,\alpha)$ corresponding to~$\mathcal{K}_c(M)$.

As a corollary we get:

\begin{thm}
	Given a typable term $M$ in~$\Fat$ it is decidable whether it has a~$c$-non-redundant typing.
\end{thm}

Given a term which admits a~$c$-non-redundant typing, the proof of the theorem above furnishes us with a method for type inference yielding an adequate, minimal $c$-non-redundant typing. In particular it furnishes us a method of type inference for typable terms.

 Note that when $c$ chooses all the polymorphic subterms of a term $M$ then $c$-non-redundancy coincides with non-redundancy (cf. Definition 4.1).

In the following subsection we give concrete examples on how we check for non-redundant typability and how we can find minimal typings whose values on terminals have only essential places.

\subsection{Examples}

In order to give some practical examples we consider a multi-sorted first-order language with equality endowed with a single binary predicate
$Occ$ and two sorts $TypeVariable$ and $TypeScheme$. $Occ$ has signature $TypeVariable \times TypeScheme$.
The occurrence theory $\mathcal{C}$ includes, besides the standard axioms for equality, the following \emph{occurrence axioms}:

\[Occ(X,\sigma\rightarrow\tau) \leftrightarrow Occ(X,\sigma)\vee Occ(X, \tau)\]

\[Occ(Y,\forall X. \sigma) \leftrightarrow Occ(Y, \sigma) \text{$\enspace$ for~$Y\neq X$}\]

\[\neg Occ(X,\forall X. \sigma)\]
\[\neg Occ(X, \sigma[Y/X]) \text{$\enspace$ for~$X\neq Y$}\]
\[Occ(X,\sigma) \leftrightarrow Occ(X, \sigma [Z /Y]) \text {$\enspace$ for~$X\neq Y$ and $Z\neq X$ } \]

\[Occ(X,\sigma) \rightarrow Occ(Y, \sigma [Y / X]  )    \]

\[ \sigma [Y/X] = \sigma [Z/X] \wedge  Occ(X,\sigma)\rightarrow Y = Z \]


If type scheme variables are interpreted as types and we assume that for all the resulting expression $A[Y/X]$, $Y$ is free for~$X$ in~$A$,   then these axioms are clearly sound for the ordinary notion of free occurrence of a variable in a type.

Consider the example of the beginning of this  section: a term $M$ containing subterms of the form $N(xX)$ and $N(xY)$.

Let $ x:\alpha_3$, $xX:\alpha_2$, $xY:\alpha_1$. Put $X_1 = \mathcal{V}(xX)$ and $X_2 = \mathcal{V}(xY)$.
The associated multisystem, after applying some rules, will include  $\alpha_1 = \alpha_2$. It will also include
\begin{align*}
  \alpha_3 & = \forall X_1 \alpha_3^{X_1}  \\
  \alpha_3 & = \forall X_2 \alpha_3^{X_2} \\
  \alpha_2 & = \alpha_3^{X_1}[X/X_1]\\
  \alpha_1 & = \alpha_3^{X_2}[Y/X_2]
\end{align*}
Also $Occ(X_1 ,  \alpha_3^{X_1}  ) $ and  $Occ(X_2 ,  \alpha_3^{X_2})$.

Applying (Quant) to the first two equations yields:
\begin{align*} \alpha_3^{X_1} & = \alpha_3^{X_2}[X_1/X_2]\\
  \alpha_3^{X_2} & = \alpha_3^{X_1}[X_2/X_1]
\end{align*}
Hence
\begin{align*}
  \alpha_2 & = \alpha_3^{X_2}[X_1/X_2][X/X_1] \\
  \alpha_1 & = \alpha_3^{X_2}[Y/X_2]
\end{align*}
Since $X_1$ does not occur in~$\alpha_3^{X_2}$ we get
\begin{align*}
  \alpha_2 & = \alpha_3^{X_2}[X/X_2]\\
  \alpha_1 & = \alpha_3^{X_2}[Y/X_2]
\end{align*}
Hence, since $\alpha_2 = \alpha_1$ by the last occurrence axiom we get that $X =Y$, a contradiction. \\

Consider the term $(y(xX))(xY)$.
Set $x:\alpha_1$,  $xX:\alpha_2$, $xY:\alpha_3$,  $y(xX):\alpha_4$, $(y(xX))(xY):\alpha_5$,  $y:\alpha_6$ and
$X_1 = \mathcal{V}(xX)$, $X_2 = \mathcal{V}(xY)$.
The associated multisystem contains:
\begin{align*}
  \alpha_6  & =  \alpha_2 \to \alpha_4 \\
\alpha_4  & =  \alpha_3\to \alpha_5  \\
 \alpha_1 & = \forall X_1. \alpha_1^{X_1} \\
 \alpha_1 & = \forall X_2. \alpha_1^{X_2} \\
 \alpha_2 & = \alpha_1^{X_1} [X/ X_1 ] \\
\alpha_3 & = \alpha_1^{X_2} [Y/ X_2 ] \end{align*}
\noindent and also $Occ(X_1, \alpha_1^{X_1})$, $Occ(X_2, \alpha_1^{X_2})$.

Applying (Quant) to the third and fourth equations yields:
\begin{align*}
  \alpha_1^{X_1} & = \alpha_1^{X_2}[X_1/X_2] \\
\alpha_1^{X_2} & = \alpha_1^{X_1}[X_2/X_1] \end{align*}
Hence $\neg Occ(X_2, \alpha_1^{X_1})$.

Let us put $ \alpha_1^{X_1} = X_1$. Then $\alpha_1 = \forall X_1. X_1$ (so $x: \forall X_1.X_1$), $\alpha_2 = X, \alpha_3 = Y$ and $\alpha_6 = X \rightarrow Y \rightarrow \alpha_5$ (so
$y:  X \rightarrow Y \rightarrow \alpha_5$).\\

Consider now $(xX)Y$. 

Let $x:\alpha_1, xX:\alpha_2, (xX)Y:\alpha_3$, $X_1 = \mathcal{V}(xX), X_2 = \mathcal{V}((xX)Y)$.

The associated multisystem contains:
\begin{align*}
  \alpha_3 & = \alpha_2^{X_2}[Y/X_2]\\
  \alpha_2 & = \alpha_1^{X_1}[X/X_1]
\end{align*}
\noindent with $Occ(X_2, \alpha_2^{X_2})$ and $Occ(X_1, \alpha_1^{X_1})$.  
Since the associated multisystem contains:
\begin{align*}
  \alpha_1 & = \forall X_1. \alpha_1^{X_1}\\
\alpha_2 & = \forall X_2. \alpha_2^{X_2} \end{align*}
\noindent we get 
\[\alpha_1^{X_1}[X/X_1] = \forall X_2. \alpha_2^{X_2} \]
Occurrence axioms yield $Occ(X,\alpha_1^{X_1}[X/X_1] )$ and thus
$Occ(X, \alpha_2^{X_2})$. We also have $Occ(X_2, \alpha_2^{X_2})$.
Also it is clear that  $\neg Occ(X_1, \alpha_2^{X_2})$. Take $
\alpha_2^{X_2} = X_2 \rightarrow X$.
Since \[\alpha_1^{X_1}[X/X_1] = \forall X_2. X_2 \rightarrow X \]
and $Occ(X_1, \alpha_1^{X_1})$ we must have $\alpha_1^{X_1} = \forall X_2.X_2\rightarrow X_1$ and hence
\[x: \forall X_1. \forall X_2. X_2 \rightarrow X_1\]
\noindent is a non-redundant typing.\\

Another example  is~$\Lambda X.(xY)$. 

Put $x:\alpha_1, xY:\alpha_2, \Lambda X.(xY):\alpha_3$ and $X_1 = \mathcal{V}(xY)$.
Then $\alpha_3 = \forall X. \alpha_2$, $\alpha_2 = \alpha_1^{X_1}[Y/X_1]$ and $\alpha_1 = \forall X_1. \alpha_1^{X_1}$ with 
$Occ(X_1, \alpha_1^{X_1})$, $Occ(X, \alpha_2)$ and $\neg Occ(X, \alpha_1)$.

Substitution yields $Occ(X,\alpha_1^{X_1}[Y/X_1] )$ and $\neg Occ(X,\forall X_1. \alpha_1^{X_1}  )$ and the occurrence axioms imply the contradiction
$Occ(X,\alpha_1^{X_1})$ and $\neg Occ(X,\alpha_1^{X_1})$.\\

A more interesting example: $(\Lambda X.\lambda x.x)Y$. In this case we get the non-redundant typing $x:X$ as can easily be checked.

\section{Final Remarks}

\begin{rem} 
	
A procedure for type inference will take as input a typable term and generate a typing  (or class of typings which is general enough in some sense).	
By the results of Section 4 we see that the output of such a procedure will
involve not only specifying occurrence (or non-occurrence) constraints for variables and terminal type scheme variables but also will involve a bound on the number of occurrences that may be required to obtain a valid typing.

	For Lemma 4.14 in the case in which $c$ does not choose any polymorphic subterms,
	that is, $\mathcal{K}_c(M) = \mathcal{B}(M)$ does not contain any positive constraints,  we obtain a method of type inference for  $\Fat$ in Polymorphic Curry style.
	
 Further work will involve a detailed analysis of the complexity of the procedure involved, preferably in the context of an actual implementation.
	
\end{rem}

\begin{rem}

In a recent work~\cite{pist} Pistone and Tranchini  proved that typability is decidable for Curry style $\Fat$. 
They do this by means of an alternative approach  which involves reducing typability to the type checking problem (TC).
To show that TC is decidable, they employ a decidable restriction, modeled after $\Fat$,  of the general undecidable second-order unification problem. The restricted problem is shown to be decidable by  reduction to the first-order case  where the problem of finding variable cycles is decidable.
 Type checking is reduced to such a restricted second-order unification problem.
 Our present work on the other hand analyzes typability not only for Curry style terms but also for Polymorphic Curry style terms and most of the paper is focused on obtaining a procedure for type inference embodying possible constraints.

The question we address in this paper of finding non-redundant typings for Polymorphic Curry terms can  be posed also for Curry terms. Note that the question of the typability of a Curry term $M$ can be seen as involving the analysis of possible Polymorphic Curry terms $N$ such that  $[N] = M$.  An assumption variable $x$ occurring in a Curry term $M$ is called \emph{non-applying} if it does not occur in a subterm of the form $xy$. We call a Polymorphic Curry term $N$ such that $[N] =M$   \emph{non-trivial} for~$M$ if~$N$ is not equal to~$M$ or obtained from $M$ by replacing non-applying assumption variables $x$ in~$M$ by~$xX$ for some type variable $X$. A typing for such an~$N$ is called a \emph{non-trivial} typing of~$M$.
For the Curry system the  following important questions can be asked:

\begin{enumerate}
\item Given a Curry term $M$ is it decidable whether $M$ has a non-trivial non-redundant typing?\footnote{This is non trivial. For example given a Curry term $M$, one cannot simply consider $\Lambda X.M$ for a certain $X$ as the example $(\lambda x. xy)z$ shows.}

\item  In searching for a non-trivial non-redundant typing for~$M$  do we  need to consider only a finite number (bounded by the complexity of~$M$) of Polymorphic Curry terms $N$ such that $[N] = M$?
\item for~$M$ having a non-trivial non-redundant typing is there a procedure to generate such  typings?
\end{enumerate}

It seems plausible that question (2) can be answered in the affirmative, specially if we consider arguments based on strong normalisation as in Remark 5.3.
In this case it is readily seen that the results of the present work allow us to answer (1) in the affirmative and that we obtain an effective procedure for question (3).

We also observe that in an implementation of~$\Fat$  in  actual programming languages the questions above are relevant. Redundant typing can be seen as problematic in terms of being computationally insignificant and a waste of resources.

\end{rem}

\begin{rem}

We mentioned in the above remark that TC for~$\Fat$ in the Curry system  has been shown to be decidable in~\cite{pist} using a second-order unification technique.  It seems there could be an alternative approach (perhaps computationally more efficient) based on the Polymorphic Curry style and the fact that $\Fat$ has strong normalisation. Given an environment $\Gamma$, a type $A$ and a Curry term $M$, the above problem is equivalent to finding a Polymorphic Curry term $N$ such that $[N]=M$ and $\Gamma\vdash N: A$.

The number of instantiations in a derivation of~$N$  is clearly bounded by the number of universal abstractions in the derivation and the total number of universal abstractions in~$\Gamma$,  so we
only need to study the bounds of these. But what about series of repeated applications of INST and GEN on the same variable?
It would seem that we need normalisation arguments like~\cite{undec} [Lemma 3.2 or Lemma 3.4]. Also

\begin{lem}
	Suppose $\Gamma\vdash_{\Fat} M:A$ where $M$ is a Polymorphic Curry term and let $M$ contain a universal abstraction on
	$X$, with $X$ not occurring in~$\Gamma$ or in~$A$. Then there is a term $M'$ such that $[M] = [M']$, $\Gamma\vdash_{\Fat} M':A$
	and $M'$ has one less occurrence of universal abstractions on~$X$.
\end{lem}

In the case in which $\Gamma$ and $A$ are canonical (in the terminology of~\cite{undec}), that is, have no redundant quantifiers,
this  follows from~\cite{undec} Lemma 3.11.

It seems plausible that using this lemma and normalisation ($\Fat$ has strong normalisation~\cite{FerreiraFerreira13}) we can show that
there are only a finite number of possible terms $N$ with $[N] = M$ which we need to consider to decide if~$\Gamma \vdash N:A$.  Of course if~$M$ is not normal then there can be no normal $N$ such that $[N] = M$. However we could restrict reduction  to polymorphic redexes only. In the derivation of such reduced terms a INST can never come after a GEN.\@ Also by canonicity sequences of GENs must be of length
bounded by the structure  of~$\Gamma$. Sequences of INST are bounded by the number of previous GENs. Hence the number of combinations of INST and GEN is finite and bounded by the structure of~$\Gamma$ and $A$ and thus  we need only consider a finite number of terms $N$ such that $[N] = M$ and the problem above is very likely decidable. We intend to proceed this line of research in future work. \end{rem}

\section*{Acknowledgments}

The authors are grateful to the anonymous referees for useful comments and suggestions regarding the preliminary versions of this paper.  In particular they acknowledge the idea of working with $c$-non-redundancy.

\bibliography{bibrefs}
\bibliographystyle{alphaurl}

\end{document}